\documentclass[12pt]{article}
\usepackage[mathscr]{eucal}
\usepackage{amssymb,latexsym}
\usepackage{verbatim}
\usepackage{amsmath}
\usepackage{amsthm}
\usepackage{enumerate}
\usepackage{authblk}
\usepackage{color}
\usepackage{url}

\usepackage[normalem]{ulem}

\newtheorem{thm}{Theorem}
\newtheorem{prop}[thm]{Proposition}

\renewcommand\l{\lambda}

\renewcommand\d{\partial}

\renewcommand\b{\beta}

\newcommand\ric{{\rm Ric}}
\renewcommand\l{\lambda}
\newcommand\g{\gamma}

\renewcommand\a{\alpha}

\newcommand\beq{\begin{equation}}
\newcommand\eeq{\end{equation}}
\newcommand\ben{\begin{enumerate}}
\newcommand\een{\end{enumerate}}
\newcommand\bit{\begin{itemize}}
\newcommand\eit{\end{itemize}}

\newcommand{\R}{\mathbb R}

\newcommand{\pd}{\partial}

\newcounter{mnotecount}

\title{
Existence of CMC Cauchy surfaces from a spacetime curvature condition
}

\author{Gregory J. Galloway\thanks{Research partially supported by NSF grant DMS-1710808.} }
\author{Eric Ling}
\affil{Department of Mathematics
\\ University of Miami }

\begin{document}
\date{}
\maketitle
\vspace{.2in}

\begin{abstract}  

In this note we present a result establishing the existence of a compact CMC Cauchy surface from a curvature condition related to the strong energy condition.

\end{abstract}

\section{Introduction}
\label{intro}

Constant mean curvature (CMC) spacelike hypersurfaces have played an important role in mathematical general relativity.  In particular, as is well-known, the problem of finding solutions to the Einstein constraint equations  is made much simpler by assuming CMC initial data.  There are also many known advantages for solving the Einstein evolution equations if one works in CMC gauge, which gives rise to a CMC foliation.  Solving the Einstein equations by this approach usually requires, to begin with, a CMC initial data hypersurface (see e.g. \cite{AM,RS}).

In the recent paper \cite{DH}, Dilts and Holst review the issue of the existence of CMC slices in globally hyperbolic spacetimes with compact Cauchy surfaces.   As discussed in \cite{DH}, most such existence results ultimately rely on barrier methods.  However, a well-known example of Bartnik \cite{Bartnik} shows that not all cosmological spacetimes have CMC Cauchy surfaces.  Vacuum examples were later obtained by Chru\'sciel, Isenberg and Pollack \cite{CIP} using gluing methods.  These examples share certain properties.  By examining various features of Bartnik's example, Dilts and Holst formulate several conjectures concerning the existence of CMC Cauchy surfaces.  We  do not settle any of these conjectures here.   Nevertheless, motivated by some of their considerations, we have obtained a new CMC existence result which relies on a certain spacetime curvature condition.

\medskip

\begin{thm}\label{main}
Let $(M,g)$ be a spacetime with compact Cauchy surfaces. Suppose $(M,g)$ is future timelike geodesically complete and has everywhere nonpositive timelike sectional curvatures, i.e. $K \leq 0$ everywhere. Then $(M,g)$ contains a CMC Cauchy surface.
\end{thm}

Some remarks about the curvature assumption are in order.  Recall, for any timelike $2$-plane, $T \subset T_pM$, the timelike sectional curvature $K(T)$ is given by
\begin{equation}
K(T) = -g\big(R(u,e)e,u\big) = -\langle R(u,e)e,u \rangle, 
\end{equation}
where $\{u,e\}$  is any basis for $T$ with $g(u,u) = -1$ and $g(e,e) = 1$ and $R$ is the Riemann curvature tensor. In particular, $K(T)$ is independent of the orthonormal basis chosen.  
(Our sign convention for $R$ is that of \cite{BEE} and opposite that of \cite{ON}.)
Standard analysis of the Jacobi equation shows that $K \le 0$ physically corresponds to attractive tidal forces; i.e. it describes gravitational attraction in the strongest sense.  

The Ricci tensor evaluated on a unit timelike vector $u \in T_pM$ can be expressed as minus the sum of timelike sectional curvatures.  Specifically, let $\{u,e_1, \dotsc, e_n\}$ be an orthonormal basis for $T_pM$ with $g(u,u) = - 1$. Let $T_i \subset T_pM$ be the timelike plane spanned by $\{u,e_i\}$. Then 
\begin{equation}\label{Ric for K}
\text{Ric}(u,u) = \sum_{i = 1}^n \langle R(u,e_i)e_i,u\rangle = -\sum_{i = 1}^n K(T_i).
\end{equation}
In particular the assumption of nonpositive timelike sectional curvatures implies the strong energy condition, $\ric(U,U) \ge 0$ for all timelike vectors $U$.  As shown in Section~3, for FLRW spacetimes,  the assumption of nonpositive timelike sectional curvatures is equivalent to the strong energy condition.  In particular, sufficiently small perturbations of FLRW spacetimes which obey the strong energy condition {\it strictly} will have negative timelike sectional curvatures.

\medskip

Since the assumption of nonpositive timelike sectional curvatures implies the strong energy condition, one is naturally led to formulate the following conjecture. 

\medskip

\noindent{\bf Conjecture.} \emph{Let $(M,g)$ be a spacetime with compact Cauchy surfaces. If $(M,g)$ is future timelike geodesically complete and satisfies the strong energy condition, i.e. $\emph{Ric}(U,U) \geq 0$ for all timelike $U$, then $(M,g)$ contains a CMC Cauchy surface.}

\medskip

The conjecture, if correct, is not likely to be easy to prove.  In particular, it would settle the Bartnik splitting conjecture  \cite[Conjecture 2]{Bartnik} in the affirmative; see \cite[Corollary 1, p.\ 621]{Bartnik}.  The conjecture above is, in a certain sense, complimentary to Conjecture 3.5 in \cite{DH}.  In this context,  it would be interesting to resolve the issue of the timelike completeness/incompleteness of the examples constructed in \cite[Section 5.1]{CIP}.

\medskip

\section{Proof of Theorem \ref{main}}

In this note we consider globally hyperbolic spacetimes $(M,g)$, $\dim M \ge 4$, with compact Cauchy surfaces. Following the convention in O'Neill \cite{ON},  we define a \emph{Cauchy surface} to be a subset $S \subset M$ which is met by every inextendible timelike curve exactly once \cite[p. 415]{ON}.

The key result underlying the proof of Theorem \ref{main} is the following fundamental CMC existence result of Bartnik \cite{Bartnik}.

\begin{thm}[Bartnik \cite{Bartnik}]\label{Bartnik thm}
Let $(M,g)$ be a globally hyperbolic spacetime with compact Cauchy surfaces that satisfies the strong energy condition. If there is a point $p \in M$ such that $M \setminus \big(I^+(p)\cup I^-(p)\big)$ is compact, then there is a regular ($C^{2,\a}$) CMC Cauchy surface passing through $p$.
\end{thm}

The proof of Theorem \ref{main} also makes use of the notion of the causal boundary of a spacetime.  Let $(M,g)$ be a globally hyperbolic spacetime. A \emph{past set} $P \subset M$ is a set such that $I^-(P) = P$. A past set $P$ is \emph{indecomposable} if $P$ cannot be expressed as the union of two past sets which are proper subsets of $P$. For any $p \in M$, the set $I^-(p)$ is an indecomposable past set. If $P$ is an indecomposable past set and there is no $p \in M$ such that $P = I^-(p)$, then $P$ is called a \emph{terminal indecomposable past} set or TIP for short. \cite[Proposition 6.8.1]{HE} shows that $P$ is a TIP if and only if there is a future inextendible timelike curve $\g$ such that $P = I^-(\g)$. The set $\mathscr{C}^+$ of all TIP's is called the \emph{future causal boundary} of $M$. The \emph{past causal boundary} $\mathscr{C}^-$ is defined time dually.

Tipler \cite{Tipler} made the very nice observation that if the future causal boundary consists of a single point (hence $I^-(\g)  = M$ for all future inextendible timelike curves~$\g$) then the key condition in Bartnik's theorem is satisfied.  In fact, Tipler discusses  somewhat more general results, requiring somewhat more involved arguments.  For the convenience of the reader, we give a simple direct proof of the following.

\smallskip

\begin{prop}[Tipler \cite{Tipler}]\label{Tipler CMC thm}
Let $(M,g)$ be a spacetime with  compact Cauchy surfaces. If $\mathscr{C}^+$ consists of a single point, then there is a point $p \in M$, sufficiently far to the future, such that $M \setminus \big(I^+(p) \cup I^-(p)\big)$ is compact.
\end{prop}

The proof is a consequence of the following two claims.

\medskip
\noindent
{\bf Claim 1. }{\it 
Let $(M,g)$ be a spacetime with compact Cauchy surfaces. If $\mathscr{C}^+$ consists of a single point, then there is a point $p \in M$  such that $\pd I^-(p)$ is a Cauchy surface.}

\proof
Let $S$ be a Cauchy surface and $\g \colon [0,\infty) \to M$ be a future inextendible timelike curve to the future of $S$. Put $p_t = \g(t)$. Since $\mathscr{C}^+$ consists of a single point, we have $M = I^-(\g)$. Therefore $\{I^-(p_t)\}_{t \in [0,\infty)}$ is an open cover of $S$. Since $S$ is compact, there is a finite subcover $\{I^-(p_{t_1}), \dotsc, I^-(p_{t_N})\}$ with $t_1 < \dotsb < t_N$. Put $p = p_{t_N}$. Note that $I^-(p_{t_i}) \subset I^-(p)$ for all $i = 1, \dotsc, N$. Therefore $S \subset I^-(p)$. Set $B = \pd I^-(p)$. Let $\lambda$ be any inextendible timelike curve. It suffices to show $\l$ intersects $B$.  Since $\l$ intersects $S$ and $S \subset I^-(p) = I^-(B)$, we know that $\l$ intersects $I^-(B)$. Also $\l$ meets $I^+(B)$ because $p \in M = I^-(\l)$.  
By the achronality of $B$, $I^-(B) \cap I^+(B) = \emptyset$. Thus, a segment of $\l$ begins in $I^-(B)$ and ends outside $I^-(B)$. It follows that $\lambda$ meets
$\d I^-(B) = 
\d I^-(p) = B$, and only at one point. 
Hence $B = \pd I^-(p)$ is a Cauchy surface.  \qed

\medskip
The  following claim was first considered  in \cite{BS} (without proof). 

\medskip
\noindent
{\bf Claim 2. }
{\it Let $(M,g)$ be a spacetime with compact Cauchy surfaces. If $\mathscr{C}^+$ consists of a single point, then $\pd I^+(p)$ is a Cauchy surface for any point $p \in M$.}

\proof  Let $p \in M$, and put $B = \pd I^+(p)$.  Let $\l\colon \R \to M$ be an inextendible timelike curve.   
 Let $S$ be a Cauchy surface through $p$.  Suppose $\lambda$ meets $S$ at the point $q$. Then, by the achronality of $S$, $q \notin I^+(B)$. However, since $\mathscr{C}^+$ consists of a single point, we have $p \in I^-(\l)$, and so $\l$ meets $I^+(B)$.  Thus, a segment of $\l$ begins outside of $I^+(B)$ and ends in $I^+(B)$.  Hence, as in the proof of Claim 1, $\l$ meets $B$. \qed

\medskip

\noindent\emph{Proof of Proposition \ref{Tipler CMC thm}.}
By Claims 1 and 2, 
there is a point $p \in M$ such that $B^- := \pd I^-(p)$ and $B^+ := \pd I^+(p)$ are compact Cauchy surfaces.  Then it follows from the compactness of `causal diamonds'  for globally hyperbolic spacetimes that $J^+(B^-) \cap J^-(B^+)$ is compact; see the corollary on p.\ 207 in \cite{HE}.   
 Moreover, it can easily be seen that $M \setminus \big(I^+(p) \cup I^-(p)\big) = J^-(B^+) \cap J^+(B^-)$: 
 The reverse inclusion holds, as otherwise one would have either  $I^+(p) \cap \d I^+(p) \neq  \emptyset$ or $I^-(p) \cap \d I^-(p) \neq \emptyset$.  For the forward inclusion, note that $B^+$ is a Cauchy surface. Therefore it separates $M$ into $I^+(B^+) = I^+(p)$ and $I^-(B^+)$. Therefore if $q \notin I^+(p)$, then $q \in B^+ \cup I^-(B^+) \subset J^-(B^+)$.\qed

\medskip

The role of the curvature assumption in Theorem \ref{main} now enters via the following proposition, the time-dual of which was recently observed in \cite{GalVega} (cf. Proposition~5.11).

\begin{prop}\label{prop for main}
Let $(M, g)$ be a spacetime with compact Cauchy surfaces and with
everywhere non-positive timelike sectional curvatures, $K \le 0$.   If $(M, g )$ is future
timelike geodesically complete then the future causal boundary 
$\mathscr{C}^+$ consists of a single element.  
\end{prop}

\begin{proof}  We comment on the proof.   If the conclusion did not hold, then 
there would exist a future inextendible timelike curve $\g$ such that $\d I^-(\g) \ne \emptyset$. 
By properties of achronal boundaries \cite{Penrose},  $\d I^-(\g)$ is an achronal $C^0$ hypersurface ruled by future inextendible null geodesics.  However, by the time-dual of \cite[Theorem 3]{EhrGal} (see also 
\cite[Theorem 14.45]{BEE}), whose proof ultimately relies on Harris's Lorentzian triangle comparison theorem \cite{Harris}, any such null geodesic would enter its own timelike future, thereby violating the achronality of $\d I^-(\g)$.
\end{proof}

Theorem \ref{main} now follows from the preceding results:

\proof[Proof of Theorem \ref{main}]
By Proposition \ref{prop for main}, the future causal boundary $\mathscr{C}^+$ consists of a single point. Hence Proposition \ref{Tipler CMC thm} shows there is a point $p \in M$ such that $M \setminus \big(I^+(p) \cup I^-(p)\big)$ is compact. Therefore the result follows from Bartnik's Theorem. 
\qed

\medskip
\noindent
{\it Remark:} In fact our arguments imply the existence of many CMC Cauchy surfaces.  Let $p$ be the point constructed in the proof of Claim 1.  Then for any point $q \in I^+(p)$,  one has $S \subset I^-(q)$.  It follows that Proposition \ref{Tipler CMC thm} holds for any $q \in I^+(p)$.  From this we can conclude that there is a CMC Cauchy surface passing through each $q \in I^+(p)$.

\section{Timelike sectional curvatures in FLRW spacetimes}

As mentioned in Section 1, we show that the assumption of nonpositive timelike sectional curvatures is equivalent to the strong energy condition for FLRW spacetimes. The result holds for arbitrary dimension, but for simplicity we will work in dimension~4.   For some related results, see \cite[Section 7]{AB}.

Let $(M,g)$ be a 4-dimensional spacetime satisfying the Einstein equations for a perfect fluid 
\begin{equation}
R_{\mu\nu} - \frac{1}{2}Rg_{\mu\nu} = 8\pi T_{\mu\nu} = 8\pi \big[(\rho + p)u_\mu u_\nu + pg_{\mu\nu}\big]
\end{equation}
where $u$ is the future pointing unit timelike vector field whose flow is the integral curves of the fluid. Tracing the Einstein equation yields $-R = 8\pi(-\rho + 3p)$. Then we can rewrite the Einstein equations in terms of the Ricci tensor.
\begin{equation}
R_{\mu\nu} = 8\pi \big[(\rho + p)u_\mu u_\nu + pg_{\mu\nu}\big] + 4\pi(\rho - 3p)g_{\mu\nu}.
\end{equation}
Let $\{u, e_1, e_2, e_3\}$ be an orthonormal basis of vector fields on $M$. Then 
\begin{align}
4\pi (\rho + 3p) &= \text{Ric}(u,u) = -K(u,e_1) - K(u,e_2) - K(u,e_3) \label{ric for u}
\\
4\pi (\rho - p)  &=\text{Ric}(e_1,e_1) = K(e_1, u) + K(e_1, e_2) + K(e_1, e_3).
\end{align}
Here $K(v,w)$ denotes the sectional curvature of the plane spanned by $v$ and $w$. Then substituting (\ref{ric for u}) into $12\pi(\rho - p) = \text{Ric}(e_1,e_1) + \text{Ric}(e_2,e_2) +\text{Ric}(e_3,e_3)$ yields
\begin{equation}
8\pi \rho = K(e_1,e_2) + K(e_1, e_3) + K(e_2, e_3).
\end{equation}

Now suppose $(M,g)$ is a FLRW spacetime.  In this case, the energy-momentum tensor necessarily takes the form of a perfect fluid (see \cite[Theorem 12.11]{ON}).  Then the local isotropy of the spatial slices implies 
\begin{equation}
K(u,e_i) = -\frac{4\pi}{3}(\rho + 3p) \:\:\:\: \text{ and } \:\:\:\: K(e_i, e_j) = \frac{8\pi}{3} \rho \: \text{ for }i \neq j.
\end{equation}
Suppose $T = \text{span}\{u, e_1\}$. Let $u' = \a u + \b e_2$ be a unit timelike vector. Then $-\a^2 + \b^2 = -1$. Let $T' = \text{span}\{u',e_1\}$. By isotropy of the spatial slices, $T'$ is completely general. Let $S = \text{span}\{e_1, e_2\}$. Note that $K(T) = K(u,e_i)$ and $K(S) = K(e_1, e_2)$. Then
\begin{align}
-K(T') &= \langle R(u',e_1) e_1, u'\rangle = \langle R(\a u + \b e_2,e_1) e_1, \a u + \b e_2\rangle \notag
\\
&= -\a^2 K(T) + 2\a \b \langle R(u,e_1)e_1, e_2\rangle + \beta^2 K(S). 
\end{align}
From formula (4) in \cite[Proposition 7.42]{ON}, we have $\langle R(u,e_1)e_1,e_2\rangle = 0$. Therefore
\begin{equation}
-K(T') = -\a^2 K(T) + \b^2 K(S).
\end{equation}
Plugging in our expressions for $K(T)$ and $K(S)$ yields
\begin{align}
-K(T') &= \a^2\frac{4\pi}{3}(\rho + 3p) + \beta^2\frac{8\pi}{3}\rho \notag
\\
&=\a^2\frac{4\pi}{3}(\rho + 3p) + (\a^2 - 1)\frac{8\pi}{3}\rho \notag
\\
&= \frac{8\pi}{3}\left[\a^2\left(\frac{3\rho}{2} + \frac{3p}{2} \right) - \rho \right] \label{sect exp}
\end{align}
Assume $K \leq 0$ everywhere holds. Then using $\a^2 \geq 1$ in (\ref{sect exp}) implies $4\pi(\rho + 3p) \geq -K(T') \geq 0$. Therefore $\rho + 3p \geq 0$.  
 Also, since $\a^2$ can take on arbitrarily large values, it follows that 
$\rho + p \geq 0$.
Conversely, suppose $\rho + p \geq 0$ and $\rho + 3p \geq 0$. The former condition along with $\a^2 \geq 1$ implies $-K(T') \geq \frac{4\pi}{3}(\rho + 3p)$. Hence the latter condition implies $K(T') \leq 0$. Since $T'$ was arbitrary, we have $K \leq 0$ everywhere. Thus

\medskip

\emph{
FLRW models have everywhere nonpositive timelike sectional curvatures $(K \leq 0)$ if and only if $\rho + p \geq 0$ and $\rho + 3p \geq 0$.}

\medskip

It is well known that the condition $\rho + p \geq 0$ and $\rho + 3p \geq 0$ is equivalent to the strong energy condition \cite[Exercise 12.10]{ON}. Therefore the strong energy condition is equivalent to assuming $K \leq 0$ everywhere for FLRW models.  Further, we see from formula \eqref{sect exp}, that if $\rho > 0$ and $p \ge 0$, then the timelike sectional curvatures are strictly negative.  
Hence, sufficiently small perturbations of spatially closed future complete FLRW spacetimes satisfying $\rho > 0$ and $p \ge 0$ will admit CMC Cauchy surfaces.  It is perhaps worth noting in this context a result of Rodnianski and Speck  (\cite[Proposition 14.4 ]{RS}), which establishes the existence of a CMC Cauchy surface in spacetimes satisfying the Einstein equations with stiff perfect fluid ($p = \rho > 0$), and having initial data sufficiently close to that of a $\mathbb{T}^3$-FLRW model with stiff perfect fluid.


\begin{thebibliography}{10}

\bibitem{AB}
S.~B. Alexander and R.~L. Bishop, \emph{Lorentz and semi-{R}iemannian spaces
  with {A}lexandrov curvature bounds}, Comm. Anal. Geom. \textbf{16} (2008),
  no.~2, 251--282.

\bibitem{AM}
L.~Andersson and V.~Moncrief, \emph{Future complete vacuum spacetimes}, The
  {E}instein equations and the large scale behavior of gravitational fields,
  Birkh\"auser, Basel, 2004, pp.~299--330.

\bibitem{Bartnik}
R.~Bartnik, \emph{Remarks on cosmological spacetimes and constant mean
  curvature surfaces}, Comm. Math. Phys. \textbf{117} (1988), no.~4, 615--624.

\bibitem{BEE}
J.~K. Beem, P.~E. Ehrlich, and K.~L. Easley, \emph{Global {L}orentzian
  geometry}, second ed., Monographs and Textbooks in Pure and Applied
  Mathematics, vol. 202, Marcel Dekker, Inc., New York, 1996.

\bibitem{BS}
R.~Budic and R.~K. Sachs, \emph{Deterministic space-times}, General Relativity
  and Gravitation \textbf{7} (1976), no.~1, 21--29, The riddle of
  gravitation--on the occasion of the 60th birthday of Peter G. Bergmann (Proc.
  Conf., Syracuse Univ., Syracuse, N. Y., 1975).

\bibitem{CIP}
P.~T. Chru{\'{s}}ciel, J.~Isenberg, and D.~Pollack, \emph{Initial data
  engineering}, Communications in Mathematical Physics \textbf{257} (2005),
  no.~1, 29--42.

\bibitem{DH}
J.~Dilts and M.~Holst, \emph{{When do spacetimes have constant mean curvature
  slices?}}, arXiv:1710.03209 (2017).

\bibitem{EhrGal}
P.~E. Ehrlich and G.~J. Galloway, \emph{Timelike lines}, Classical and Quantum
  Gravity \textbf{7} (1990), no.~3, 297.

\bibitem{GalVega}
G.~J. Galloway and C.~Vega, \emph{Hausdorff closed limits and rigidity in
  {L}orentzian geometry}, Ann. Henri Poincar\'e \textbf{18} (2017), no.~10,
  3399--3426.

\bibitem{Harris}
S.~G. Harris, \emph{A triangle comparison theorem for {L}orentz manifolds},
  Indiana Univ. Math. J. \textbf{31} (1982), no.~3, 289--308.

\bibitem{HE}
S.~W. Hawking and G.~F.~R. Ellis, \emph{The large scale structure of
  space-time}, Cambridge University Press, London, 1973, Cambridge Monographs
  on Mathematical Physics, No. 1.

\bibitem{ON}
B.~O'Neill, \emph{Semi-{R}iemannian geometry}, Pure and Applied Mathematics,
  vol. 103, Academic Press Inc. [Harcourt Brace Jovanovich Publishers], New
  York, 1983.

\bibitem{Penrose}
R.~Penrose, \emph{Techniques of differential topology in relativity}, Society
  for Industrial and Applied Mathematics, Philadelphia, Pa., 1972, Conference
  Board of the Mathematical Sciences Regional Conference Series in Applied
  Mathematics, No. 7.

\bibitem{RS}
I.~Rodnianski and J.~Speck, \emph{{Stable Big Bang Formation in Near-FLRW
  Solutions to the Einstein-Scalar Field and Einstein-Stiff Fluid Systems}},
  arxiv:1407.6298.

\bibitem{Tipler}
F.~J. Tipler, \emph{A new condition implying the existence of a constant mean
  curvature foliation}, Directions in General Relativity: Proceedings of the
  1993 International Symposium, Maryland: Papers in Honor of Dieter Brill
  \textbf{2} (1993), no.~10, 306--315.

\end{thebibliography}

\providecommand{\bysame}{\leavevmode\hbox to3em{\hrulefill}\thinspace}
\providecommand{\MR}{\relax\ifhmode\unskip\space\fi MR }
\providecommand{\MRhref}[2]{%
  \href{http://www.ams.org/mathscinet-getitem?mr=#1}{#2}
}
\providecommand{\href}[2]{#2}

\end{document}